\documentclass[12pt]{article} % December 12, 2011.
\usepackage{amsmath,amssymb,amsfonts,amsthm,amsbsy}
\usepackage{graphics}                 % Packages to allow inclusion of graphics
\usepackage{color}

\allowdisplaybreaks

\newtheorem{theorem}{Theorem}[section]

\newtheorem{remark}[theorem]{Remark}

\newcommand{\argmin}{\mathop{\rm arg\min}}

\newcommand{\ve}{\mathbf{e}}

\newcommand{\vv}{\mathbf{v}}
\newcommand{\vw}{\mathbf{w}}
\newcommand{\vg}{\mathbf{g}}

\def\I{\mathbb{ I}}

\def\R{\mathbb{ R}}

\title{On the $\ell_1$-Norm Invariant Convex $k$-Sparse Decomposition
of Signals}

\author{
 Guangwu Xu \thanks{Department of EE \& CS, University of Wisconsin-Milwaukee,
Milwaukee, WI 53211, USA; e-mail: {\tt gxu4uwm@uwm.edu}. Research
supported in part by the National 973 Project of China (No.
2013CB834205).}  and
Zhiqiang Xu \thanks{Inst. Comp. Math., Academy of Mathematics and Systems Science,
Chinese Academy of Sciences, Beijing, China;
e-mail: {\tt xuzq@lsec.cc.ac.cn}.      Zhiqiang Xu was supported  by NSFC grant 11171336, 11331012, 11021101 and National Basic Research
Program of China (973 Program 2010CB832702).}
}

\date{}

\begin{document}
\maketitle

\begin{abstract}
Inspired by an interesting idea of Cai and Zhang, we formulate
and prove the  convex $k$-sparse decomposition
of vectors that is invariant with respect to the $\ell_1$ norm.
This result fits well
in discussing compressed sensing problems under RIP, but we believe it
also has independent interest. As an application, a simple derivation of the RIP
recovery condition $\delta_k+\theta_{k,k} < 1$ is presented.
\end{abstract}

\noindent{\bf Keywords:\/} Convex $k$-sparse decomposition,
$\ell_1$ minimization, restricted isometry property, sparse recovery.

\section{Introduction}

The {\sl Restricted Isometry Property (RIP)} of Cand\`es and Tao \cite{CanTao05}
is one of the most commonly used frameworks for sparse recovery via $\ell_1$
minimization. For an $n\times p$ matrix $\Phi\in \R^{n\times p}$ and an integer $k$, $1\le k \le p$,
the {\it $k$-restricted isometry constant} $\delta_k$ is the smallest constant such that
\[
\sqrt{1-\delta_k}\|c\|_2 \le \|\Phi c\|_2 \le \sqrt{1+\delta_k}\|c\|_2
\]
for every $k$-sparse vector $c$ (namely, $c$ has at most $k$ nonzero components). If $k+k'\le p$, the
{\it $k, k'$-restricted orthogonality constant} $\theta_{k,k'}$,
is the smallest number that satisfies
\[
|\langle \Phi c, \Phi c'\rangle| \le \theta_{k, k'}\|c\|_2\|c'\|_2,
\]
for all $k$-sparse vector $c$ and  $k'$-sparse vector $c'$ with
disjoint supports.

It has been shown that $\ell_1$ minimization can recover a sparse
signal with a small or zero error under various conditions on
$\delta_k$ and $\theta_{k,k'}$, such as the condition
$\delta_k+\theta_{k,k}+\theta_{k,2k}<1$ in Cand\`es
and Tao \cite{CanTao05}, and
the condition $\delta_{k}<0.307$ of Cai, Wang and Xu \cite{CMX2}.
Recently, Cai and Zhang \cite{CZ} established a sharp condition on $\delta_k$  for
$k$-sparse recovery:
\[
\delta_{k}<\frac{1}3.
\]
In the same paper, they also proved that $\delta_{2k}<\frac{1}2$ is sufficient for $k$-sparse signal reconstruction.
Cai and Zhang developed a marvelous technique in the proof of their
results.

Inspired by the {\sl division lemma} of Cai and  Zhang \cite{CZ}, we formulate and prove the {\sl $\ell_1$-norm invariant
convex $k$-sparse decomposition} of vectors in this note. This result (Theorem
\ref{thm:2.1}) asserts that every vector is a convex combination of
 $k$-sparse vectors with invariant $\ell_1$ norm. Such decomposition fits well in treating compressed sensing
problems under RIP, as a tighter conversion between $\ell_1$-norm and $\ell_2$ norm is desired.
We shall demonstrate this by showing how to use the decomposition to
derive the sparse recovery condition $\delta_k+\theta_{k,k}<1$ of Cai and
Zhang \cite{CZ1} in a simple manner.
However, we believe
that this result is of independent interest for other applications.

After the early appearance of this note  (arXiv:1305.6021, May 2013), we learned that Cai and
Zhang \cite{CZ2} also established a similar decomposition and
using it to derive some good RIP conditions (e.g., $\delta_{2k}\le \frac{1}{\sqrt{2}}$). Using
the $\ell_1$-norm invariant convex $k$-sparse decomposition, under the tight frame sparsification,
Baker \cite{Baker} obtained the condition $\delta_{2k}\le \frac{1}{\sqrt{2}}$ for the
Dictionary-Restricted Isometry Property.

The paper is organized as follows. Section 2 presents the $\ell_1$-norm invariant
$k$-sparse convex decomposition. As a consequence of this decomposition, we prove a useful result for comparing $\ell_p$ norms. In Section 3, the $k$-sparse convex decomposition is used to
give a simple derivation of the sparse recovery condition $\delta_k+\theta_{k,k}<1$ of Cai and Zhang.

\section{Convex $k$-Sparse Decomposition}
In this section, we prove that every vector is a convex combination of
$k$-sparse vectors with invariant $\ell_1$ norm.  The formulation is inspired by the celebrated ideas from Cai and Zhang \cite{CZ}.
 We also show that the
$\ell_{\infty}$ norm of the summand vectors is well behaved. More
specifically, we have
\begin{theorem}\label{thm:2.1}
For positive integers $k\le n$, and positive constant $C$,
let $\vv\in \R^{n}$ be a vector with
\[
\|\vv\|_1\le C \quad\quad \mbox{ and} \quad\quad  \|\vv\|_{\infty}\le \frac{C}k .
\]
Then there are $k$-sparse vectors $\vw_1,\dots, \vw_M$ with
\begin{equation}\label{L1Inv}
\|\vw_t\|_1=\|\vv\|_1 \quad \mbox{ and} \quad \|\vw_t\|_{\infty}\le \frac{C}k \quad \mbox{ for } t=1,\cdots, M,
\end{equation}
such that
\begin{equation}\label{k-sparse}
\vv = \sum_{t=1}^M x_t \vw_t
\end{equation}
for some nonnegative real numbers $x_1, \dots, x_M$ with $\sum_{t=1}^Mx_t=1$.
\end{theorem}
\begin{proof}
If $k=n$, or $\vv$ is already $k$-sparse,
then there is nothing to do. Assume now $n>k$.
Without loss of generality, we may consider the case that all components of $\vv$ are positive\footnote{In this note, we will
treat a vector of $\R^n$ as a function from $\{1,\cdots, n\}$ to $\R$.} (the general case
 can be argued easily, as (\ref{k-sparse}) still holds by multiplying $-1$ to the $i$th components of both sides )
\[
\vv(1)\ge \vv(2)\ge \cdots \ge \vv(n)>0.
\]

For each $j=1,\dots,k$, let $\eta_j := \frac{C}k - \vv(j)$. Since $\sum_{j=1}^k \eta_j=C-\vv(1)-\cdots-\vv(k)\ge \vv(k+1)+\cdots+\vv(n)$, so
we have $\eta_j\ge 0$ and $\sum_{j=1}^k\eta_j>0$.
Let
\[
\lambda_i := \frac{\eta_i}{\sum_{j=1}^k \eta_j}, \qquad i = 1, 2, \ldots, k,
\]
Then $\sum_{i=1}^k\lambda_i=1$.

We shall construct $k+1$ vectors $\vg_0, \dots, \vg_k$;
each has $n-1$ nonzero components and
satisfies
$$
\|\vg_t\|_1=\|\vv\|_1 \quad \mbox{ and} \quad \|\vg_t\|_{\infty}\le \frac{C}k \quad \mbox{ for } t=0, 1,\ldots, k.
$$
 Furthermore, $\vv$ is a convex combination of  $\vg_0, \dots, \vg_k$.
In the following construction, we will use $\vv_{\{j\}}$ to denote the vector
whose $j^{th}$ component is $\vv(j)$ and other components are zero, and  use $\{\ve_1, \ldots, \ve_n\}$
to denote the canonical basis. The $k+1$ vectors $\vg_0, \dots, \vg_k$ are
\begin{eqnarray*}
&&
\vg_0=\sum_{j=1}^{k}(\vv(j)+\lambda_j\vv(n))\ve_j+\vv_{\{k+1\}}+\cdots+\vv_{\{n-1\}},\\
&&\vg_t=\sum_{\scriptstyle j=1\atop\scriptstyle j\neq t}^{k}(\vv(j)+\lambda_j\vv(n))\ve_j+(\vv(t)+\lambda_t\vv(n))\ve_n+\vv_{\{k+1\}}+\cdots+\vv_{\{n-1\}}, 1\leq t\leq k.\\
\end{eqnarray*}
Let
$$y_1:=\frac{\lambda_1\vv(n)}{\vv(1)+\lambda_1\vv(n)},\,
y_2:=\frac{\lambda_2\vv(n)}{\vv(2)+\lambda_2\vv(n)},\, \dots,
y_k:=\frac{\lambda_k\vv(n)}{\vv(k)+\lambda_k\vv(n)}.
$$
Then we see that $y_i< \lambda_i$. So, by setting $y_0=1-y_1-\cdots-y_k$, we have
$y_0> 0$ and
\[
y_0+y_1+\cdots+y_k=1.
\]
It is also straightforward to verify
\[
\vv = y_0\vg_0+y_1\vg_1+\cdots+y_k\vg_k.
\]
For example, the first component of $y_0\vg_0+y_1\vg_1+\cdots+y_k\vg_k$
is
\begin{eqnarray*}
&& y_0(\vv(1)+\lambda_1\vv(n))+y_2
(\vv(1)+\lambda_1\vv(n))+\cdots+y_k(\vv(1)+\lambda_1\vv(n))\\
&& \quad =\left(\vv(1)+\lambda_1\vv(n)\right)
(1-y_1) = \vv(1).
\end{eqnarray*}
The other requirements for $\vg_t\quad (t=0, 1, \dots, k$) are
\begin{enumerate}
\item $\|\vg_t\|_1=\|\vv\|_1$. This is certainly true.
\item $|\vg_t(i)|\le \frac{C}k$. To see this, we note that
\[
\vg_t(i) =\left\{
\begin{array}{ll} 0 &\mbox{ if } i=t \mbox{ or if } t=0, i=n;\\
\vv(i)+\lambda_i\vv(n) & \mbox{ if } 1\le i \le k
             \mbox{ and } i\neq t ;\\
\vv(i) & \mbox{ if } k<i<n;\\
\vv(t)+\lambda_t\vv(n)  & \mbox{ if } t>0, i=n.
\end{array}\right.
\]
For $1\le i \le n$,
\begin{eqnarray*}
\frac{C}k - |\vg_t(i)| &\ge& \frac{C}k -
(\vv(i)+\lambda_i\vv(n))=\eta_i - \lambda_i\vv(n)\\
&=&
\eta_i-\frac{\eta_i\vv(n)}{\sum_{j=1}^k\eta_j}
\ge
\eta_i-\frac{\eta_i\vv(n)}{\vv(k+1)+\cdots+\vv(n)}\ge 0.
\end{eqnarray*}
\end{enumerate}

If $n-1>k$, we repeat this process for each $\vg_t$, and so on, until a
$k$-sparse convex combination is reached.
\end{proof}
\begin{remark}
The proof of Theorem \ref{thm:2.1} in fact presents a method to construct the vectors ${\bf w}_t, t=1,\ldots,M$ with $M={n\choose k}$. Using this method, the time complexity to construct the $M$ vectors is $O(k^{n-k})$. It will be interesting to design efficient algorithms to construct the vectors.
\end{remark}

\section{RIP Conditions in Compressed Sensing}

In this section, we shall use the $k$-sparse convex decomposition
to describe a short proof of the following results of Cai and Zhang
\cite{CZ1}. Our proof follows an approach similar to that
in \cite{CMX1,CMX2}. We first consider the recovery of $k$-sparse signal:
\begin{theorem}\label{thm:3.1}
Let $\beta$ be a $k$-sparse signal and $y=\Phi \beta$ where $\Phi$ satisfies
\begin{equation}\label{eq:3.1}
\delta_k +\theta_{k,k}< 1.
\end{equation}
Let \[
\hat \beta = \argmin_{\gamma \in \R^p} \{\|\gamma\|_1 \; \quad \mbox{\rm
  subject to } \quad y = \Phi\gamma\}
\]
Then $\beta=\hat\beta$.
\end{theorem}
\begin{proof}
Let $h:=\hat \beta -\beta$. We need to show that $h=0$. Otherwise, we  assume $|h(1)|\ge |h(2)|\ge  \cdots \ge  |h(p)|>0$. Denote
\[
T:=\{1, 2, \cdots,  k\}, \quad S:=\{k+1,k+2, \cdots, p\},
\]
then as in \cite{CMX2}, the minimality of $\hat\beta$ yields
\[
\|h_{S}\|_1\leq \|h_{T}\|_1 ,
\]
where $h_Q=h\I_Q$ and $\I_Q$ denotes the indicator function of the set $Q$ (namely,
$\I_Q(j) =1$ if $j\in Q$ and 0 if $j\notin Q$).

From the assumption, we also have $|h_S(j)|\le \frac{\|h_T\|_1}k$  for all $j\in S$, i.e.
\[
\|h_S\|_{\infty} \le \frac{\|h_T\|_1}k .
\]
Therefore, by theorem \ref{thm:2.1}, $h_S$ can be written as
\[
h_{S}=\sum_{j=1}^M x_j \vw_j \quad \quad \mbox{ where }x_j\ge 0 \mbox{ and } \sum_{j=1}^q x_j = 1,
\]
with each $\vw_j$ is $k$-sparse and supported on $S$, and $\|\vw_j\|_1=\|h_S\|_1,
\|\vw_j\|_{\infty}\le \frac{\|h_T\|_1}k$. As $h_T$ and $\vw_j$ have disjoint supports
and $\|\vw_j\|_2\le \sqrt{k\big(\frac{\|h_T\|_1}k\big)^2}\le \|h_T\|_2$, we get
\begin{eqnarray*}
(1-\delta_k)\|h_{T}\|_2^2 &\le & \|\Phi h_{T}\|_2^2 = |\langle \Phi h_{T},
\Phi h_{S}\rangle|\le
                          \sum_{j=1}^M x_j|\langle \Phi h_{T}, \Phi \vw_j\rangle|\\
 &\le & \sum_{j=1}^M x_j\theta_{k,k}\|h_T\|_2\|\vw_j\|_2\le \sum_{j=1}^M x_j\theta_{k,k}\|h_{T}\|_2\|h_{T}\|_2\\
 &=&\theta_{k,k}\|h_{T}\|_2^2.
 \end{eqnarray*}
We have reached a contradiction. Hence $h=0$.
\end{proof}

\begin{remark}
We state the proof for the $k$-sparse signal. In fact, one also can extend the proof to  the noise case easily. In \cite{Baker},
Baker stated such a proof for the case where the signals are sparse in a redundant dictionary.
\end{remark}


\begin{thebibliography}{9}
\bibitem{Baker} C. A. Baker, A Note on Sparsification by Frames, {\sl August, 2013, http://arxiv.org/abs/1305.6021}.
\bibitem{CZ} T. Cai and A. Zhang, Sharp RIP bound for sparse signal and low-rank matrix recovery, {\sl Applied and Computational Harmonic Analysis}, 35(2013), 74-93.
\bibitem{CZ1} T. Cai and A. Zhang, Compressed sensing and affine rank minimization under restricted isometry,
{\sl IEEE Transactions on Signal Processing}, to appear.
\bibitem{CZ2} T. Cai and A. Zhang, Sparse representation of a polytope and recovery of sparse signals and low-rank matrices,
{\sl June, 2013, http://arxiv.org/abs/1306.1154}.
\bibitem{CMX1} T. Cai, L. Wang, and G. Xu, Stable recovery of sparse signals and an oracle inequality, {\sl IEEE Transactions on Information Theory}, 56(2010), 3516-3522.
\bibitem{CMX2}T. Cai, L. Wang, and G. Xu, New Bounds for restricted isometry constants, {\sl IEEE Transactions on Information Theory}, 56(2010), 4388-4394.
\bibitem{CanTao05} E. J. Cand\`es and T. Tao, Decoding by linear programming, {\sl IEEE Trans. Inf. Theory}, 51(2005)
4203-4215.
\end{thebibliography}
\end{document}